\documentclass[sigconf]{acmart}

\usepackage{algorithm}
\usepackage{algorithmic}
\usepackage{url}
\usepackage{multirow}
\usepackage{hyperref}

\copyrightyear{2017}
\acmYear{2017}
\setcopyright{acmcopyright}
\acmConference{ISSAC '17}{July 25-28, 2017}{Kaiserslautern, Germany}
\acmPrice{15.00}
\acmDOI{http://dx.doi.org/10.1145/3087604.3087608}
\acmISBN{978-1-4503-5064-8/17/07}

\begin{document}

\title{Functional Decomposition Using Principal Subfields}

\author{Luiz E. Allem}
\affiliation{%
  \institution{Univ. Federal do Rio Grande do Sul}
  \streetaddress{Av. Bento Gon\c calves, 9500}
  \city{Porto Alegre} 
  \state{RS} 
  \postcode{91509-900}
}
\email{emilio.allem@ufrgs.br}

\author{Juliane G. Capaverde}
\affiliation{%
  \institution{Univ. Federal do Rio Grande do Sul}
  \streetaddress{Av. Bento Gon\c calves, 9500}
  \city{Porto Alegre} 
  \state{RS} 
  \postcode{91509-900}
}
\email{juliane.capaverde@ufrgs.br}

\author{Mark van Hoeij}
\affiliation{%
  \institution{Florida State University}
  \streetaddress{211 Love Building}
  \city{Tallahassee} 
  \state{FL} 
  \postcode{32306}
}
\email{hoeij@math.fsu.edu}

\author{Jonas Szutkoski}
\affiliation{%
  \institution{Univ. Federal do Rio Grande do Sul}
  \streetaddress{Av. Bento Gon\c calves, 9500}
  \city{Porto Alegre} 
  \state{RS} 
  \postcode{91509-900}
}
\email{jonas.szutkoski@ufrgs.br}

\begin{abstract}
Let $f\in K(t)$ be a univariate rational function. It is well known that any non-trivial decomposition $g\circ h$, with $g,h\in K(t)$, corresponds to a non-trivial subfield $K(f(t))\subsetneq L\subsetneq K(t)$ and vice-versa. In this paper we use the idea of \textit{principal subfields} and fast subfield-intersection techniques to compute the subfield lattice of $K(t)/K(f(t))$. This yields a Las Vegas type algorithm with improved complexity and better run times for finding \textit{all} non-equivalent complete decompositions of $f$.
\end{abstract}
%
%
%


\maketitle

\renewcommand{\shortauthors}{L. E. Allem et al.}

\begin{acks}
This work is part of the doctoral studies of Jonas Szutkoski, who acknowledges the support of \grantsponsor{Capes}{CAPES}{http://www.capes.gov.br/}. Mark van Hoeij was supported by the \grantsponsor{NSF}{National Science Foundation}{https://www.nsf.gov/} under Grant No.:\grantnum{NSF}{1618657}.
\end{acks}

\section{Introduction}

The problem of finding a decomposition of a rational function $f\in K(t)$ has been studied by several authors. We highlight the work of \cite{Z1991}, who gave the first polynomial time algorithm that finds (if it exists) a single decomposition of $f$. In \cite{ALONSO1995}, an exponential time algorithm was given that computes all decompositions of $f$ by generalizing the ideas of \cite{Barton1985} for the polynomial case. More recently, \cite{Ayad2008} have presented improvements on the work of \cite{ALONSO1995}, though the complexity is still exponential on the degree of $f$.

The particular case of polynomial decomposition has long been studied. As far as the authors' knowledge goes, the first work on polynomial decomposition is from \cite{Ritt1992}, which presented a strong structural property of polynomial decompositions over complex numbers. In \cite{Barton1985}, two (exponential time) algorithms are presented for finding the decompositions of a polynomial over a field of characteristic zero. Some simplifications are suggested in \cite{Alagar1985, ALONSO1995}. In \cite{Kozen1989}, the first polynomial time algorithm is given, which works over any commutative ring containing an inverse of $\deg (g)$. Further improvements are presented in \cite{Gathen1990t, Gathen1990w}. More recently, \cite{Blankertz2014} presented a polynomial time algorithm that finds all \textit{minimal decompositions} of $f$, with no restrictions on $\deg (g)$ or the characteristic of the field.

\textit{Univariate Functional Decomposition} (either rational function or polynomial) is closely related to the subfield lattice of the field extension $K(t)/K(f(t))$ (see Theorem \ref{theo:main} below). However, in general, the number of subfields is not polynomially bounded and algorithms for finding \textit{all} complete decompositions can suffer a combinatorial explosion. In this work, we try to improve the non-polynomial part of the complexity. In order to achieve this, we make use of the so-called \textit{principal subfields}, as defined in \cite{vanHoeij2013}. 

Let $f(t)=p(t)/q(t)\in K(t)$, $n=\max\{\deg(p),\deg(q)\}$ and $\nabla_f:=p(x)q(t)-p(t)q(x)\in K[x,t]$. Assuming we are given the factorization of $\nabla_f$, using fast arithmetic and fast subfield intersection techniques (see \cite{Szutkoski2016}), we can compute the subfield lattice of $K(t)/K(f(t))$ with an expected number of \[\tilde{\mathcal{O}}(rn^2) \text{ field operations plus }\tilde{\mathcal{O}}(mr^2) \text{ CPU operations,}\]  where $m$ is the number of subfields of $K(t)/K(f(t))$ and $r\leq n$ is the number of irreducible factors of $\nabla_f$ (see Corollary \ref{cor:complex}). 
This approach has the following improvements:
\begin{itemize}
\item Better complexity: our algorithm does not depend exponentially on $r$ as previous methods (e.g., \cite{Ayad2008}), only on the number $m$ (usually $m\ll 2^r$). Furthermore, the non-polynomial part of the complexity is reduced to CPU operations.
\item Better run times: an implementation in Magma shows the efficiency of our algorithm when compared to \cite{Ayad2008}.
\item Better complexity for polynomial decomposition (especially in the wild case): given $f(t)\in \mathbb{F}_q[t]$, we can find all \textit{minimal decompositions} of $f$ with an expected number of $\tilde{\mathcal{O}}(rn^2)$ operations in $\mathbb{F}_q$ plus the cost of factoring $\nabla_f=f(x)-f(t)\in \mathbb{F}_q[x,t]$, where $r$ is the number of irreducible factors of $\nabla_f$. See Remark \ref{rem:poly}.
\end{itemize}

As previous methods, our algorithm requires the factorization of a bivariate polynomial over $K$ of total degree at most $2n$, where $n$ is the degree of $f$. 

\subsection{Roadmap}
In Section \ref{sec:2}, we recall some basic definitions and results about rational function decomposition. Let $K$ be a field and let $f\in K(t)$ be a rational function. In Section \ref{sec:3}, we give a description of the principal subfields of the extension $K(t)/K(f)$. Every subfield of a finite separable field extension corresponds to a unique partition on the set of irreducible factors of the minimal polynomial of this extension. In Section \ref{sec:4}, we show how one can compute this partition for every principal subfield. This allows us to compute the subfield lattice of $K(t)/K(f(t))$ efficiently. Finally, in Section \ref{sec:5}, we show how one can use these partitions to compute all decompositions of $f$. Some timings comparing our algorithm with \cite{Ayad2008} are also given.

\subsection{Complexity model}

Throughout this paper, field operations $(+, -, \times, \div)$ and the equality test are assumed to have a constant cost. Given polynomials $f,g\in K[x]$ of degree at most $n$, we can compute their product (and the remainder of $f$ divided by $g$) with $\mathcal{O}(M(n))$ field operations. We recall that $M$ is \textit{super-additive}: $M(n_1)+M(n_2)\leq M(n_1+n_2)$ (see \cite{MCA}, Chapter 8.3). If $f\in K[x]$ is irreducible with degree $n$, then arithmetic in $K[x]/(f)$ costs $\mathcal{O}(M(n))$ operations in $K$ (see \cite{MCA}, Chapter 9). Furthermore, the greatest common divisor of two polynomials $f,g$ of degree $\leq n$ costs $\mathcal{O}(M(n)\log n)$ field operations (see \cite{MCA}, Chapter 11). 
Finally, given a linear system $\mathcal{S}$, with $m$ equations in $r$ variables, we can compute a basis of solutions of $\mathcal{S}$ with $\mathcal{O}(mr^{\omega-1})$ field operations (see \cite{bini}, Chapter 2), where $2<\omega \leq 3$ is a \textit{feasible matrix multiplication exponent} (see \cite{MCA}, Chapter 12) . 

\section{Basic Definitions}\label{sec:2}
Let $K$ be an arbitrary field and let $K(t)$ be the function field over $K$. Let $\mathbb{S}=K(t)\backslash K$ be the set of non-constant rational functions and let $f=f_n/f_d\in \mathbb{S}$ be a rational function with $f_n, f_d\in K[t]$ coprime. The \textit{degree} of $f$ is defined as $\max \{ \deg(f_n), \deg(f_d) \}$ and denoted by $\deg(f)$. The set $\mathbb{S}$ is equipped with a structure of a monoid under composition. The $K$-automorphisms of $K(t)$ are the fractional transformations $u=(ax+b)/(cx+d)$ such that $ad-bc\neq 0$. The group of automorphisms is isomorphic to $PGL_2(K)$ and also to the group of units of $\mathbb{S}$ under composition.

An element $f\in K(t)$ is \textit{indecomposable} if $f$ is not a unit and $f=g \circ h$ implies $g$ or $h$ is a unit. Otherwise, $f$ is called \textit{decomposable}. If $f$ is decomposable with $f=g\circ h$, then $h$ (resp. $g$) is called the \textit{right} (resp. \textit{left}) \textit{component} of the decomposition $g\circ h$. Furthermore, a decomposition $f=g\circ h$ is \textit{minimal} if $h$ is indecomposable and a decomposition $f=g_m \circ \cdots \circ g_1$ is \textit{complete} if all $g_i$ are indecomposable.

 It is well known by L\"uroth's Theorem that if $K \subsetneq L \subseteq K(t),$ then there exists $h\in \mathbb{S}$ such that $L=K(h)$ (a proof can be found in \cite{ModernAlgebra}). The rational function $h$ is not unique however, $K(h)=K(h')$, if and only if, there exists a unit $u\in \mathbb{S}$ such that $h'=u \circ h$.
 As in \cite{Ayad2008}, we define the \textit{normal form} of a rational function $f\in \mathbb{S}$.
\begin{definition}\label{def:normal}
A rational function $f=p/q\in \mathbb{S}$ is in \textit{normal form} or \textit{normalized} if $p, q\in K[t]$ are monic, coprime, $p(0)=0$ and either $deg(p)>\deg(q)$ or $m:=\deg(p)<\deg(q)=:n$ and $q=t^n+q_{n-1}t^{n-1}+\cdots+q_0,$ with $q_m=0$.
\end{definition}

Given $f\in \mathbb{S}$, there exists a unique normalized $\hat{f}\in \mathbb{S}$ such that $K(f)=K(\hat{f})$ (\cite{Ayad2008}, Proposition 2.1). Hence, if $\mathcal{N}_K$ is the set of all normalized rational functions over $K$, then there exists a bijection between $\mathcal{N}_K$ and the set of fields $L$ such that $K \subsetneq L \subseteq K(t)$. 
In particular, there is a bijection between normalized rational functions $h\in \mathbb{S}$ such that $f=g\circ h$, for some $g\in \mathbb{S}$, and the fields $L=K(h)$ such that $K(f)\subseteq L\subseteq K(t)$.

\begin{definition}
For a rational function $g=g_n/g_d\in \mathbb{S}$, with $g_n, g_d\in K[t]$ coprime, define $\nabla_g(x,t):=g_n(x)g_d(t)-g_n(t)g_d(x)\in K[x,t]$ and $ \Phi_g(x):=g_n(x)-g(t)g_d(x)\in K(g)[x]. $ A bivariate polynomial $a(x,t)\in K[x,t]$ is called \textit{near-separate} if $ a(x,t)=\nabla_{g}(x,t)$, for some $g\in K(t)$.
\end{definition}

In this work, we assume that $f$ is such that $\Phi_f$ is monic. If this is not the case, we can find a unit $u \in K(t)$ such that $\tilde{f}:=u\circ f$ and $\Phi_{\tilde{f}}$ is monic. Decomposing $f$ is equivalent to decomposing $\tilde{f}.$

\newtheorem{remark}{Remark}
\begin{remark}\label{rem:fact}
Let $f\in K(t)$ of degree $n$ and let $G_1,\ldots, G_r$ be the irreducible factors of $\nabla_f\in K[x,t]$. Let $m_1,\ldots, m_r\in K[t]$ be the leading coefficients of $G_1,\ldots, G_r$ w.r.t. $x$. Then $m_1 \cdots m_r = f_d(t)$ and $F_i:=G_i/m_i\in K(t)[x]$ are monic, irreducible and $\nabla_f/f_d(t)=\Phi_f(x)=F_1\cdots F_r.$ In particular, if the exponents of $t$ in $G_i$ are bounded by $d_i$, then $\sum d_i =n$.

\end{remark}

The following theorem is the key result behind all \textit{near-separate based} rational function decomposition algorithms, such as \cite{ALONSO1995} and \cite{Ayad2008} (see also \cite{Barton1985} for the polynomial case).

\begin{theorem}[\cite{ALONSO1995}, Proposition 3.1]\label{theo:main}
Let $f,h\in \mathbb{S}$ be rational functions. The following are equivalent:
\begin{itemize}
\item[$a)$] $K(f)\subseteq K(h)\subseteq K(t)$.
\item[$b)$] $f=g\circ h$, for some $g\in \mathbb{S}$.
\item[$c)$] $\nabla_{h}(x,t)$ divides $\nabla_{f}(x,t)$ in $K[x,t]$.
\item[$d)$] $\Phi_h(x)$ divides $\Phi_f(x)$ in $K(t)[x]$.
\end{itemize}
\end{theorem}

If $G_1,\ldots, G_r$ are the irreducible factors of $\nabla_f$ over $K[x,t]$, then the product of any subset of $\{G_1,\ldots, G_r\}$, which is a near-separate multiple of $x-t$, yields a right component $h$ and hence, a decomposition $f=g\circ h.$ Many authors use this approach to compute all decompositions of $f$: factor $\nabla_f$ and search for near-separate factors (see  \cite{ALONSO1995, Ayad2008, Barton1985}). However, this approach leads to exponential time algorithms due to the number of factors we have to consider.

\section{Principal Subfields}\label{sec:3}

In this section we use the idea of \textit{principal subfields} to compute the subfield lattice of $K(t)/K(f)$. By Theorem \ref{theo:main}, this gives us all complete decompositions of $f$. Principal subfields and fast field intersection techniques (see \cite{Szutkoski2016}) allow us to improve the non-polynomial part of the complexity.

\subsection{Main Theorem}
Let $K/k$ be a separable field extension of finite degree $n$. A field $L$ is said to be \textit{a subfield of} $K/k$ if $k\subseteq L \subseteq K$. It is well known that the number of subfields of $K/k$ is not polynomially bounded in general. However, we have the following remarkable result from \cite{vanHoeij2013}: 

\begin{theorem}\label{theo:main3}
Given a separable field extension $K/k$ of finite degree $n$, there exists a set $\{L_1,\ldots, L_r\}$, with $r\leq n$, of subfields of $K/k$ such that, for any subfield $L$ of $K/k$, there exists a subset $I_L\subseteq \{1,\ldots, r\}$ with \[ L=\bigcap_{i\in I_L} L_i. \]
\end{theorem}

 The subfields $L_1,\ldots,L_r$ are called \textit{principal subfields} of the extension $K/k$ and can be obtained as the kernel of some application (see \cite{vanHoeij2013}). Instead of directly searching for all subfields of a field extension, which leads to an exponential time complexity, principal subfields allow us to search for a specific set of $r\leq n$ subfields, a polynomial time task. 
 
 By Theorem \ref{theo:main3}, the non-polynomial part of the complexity of computing the subfield lattice is then transfered to computing all intersections of the principal subfields. However, according to \cite{Szutkoski2016}, each subfield of $K/k$ can be uniquely represented by a partition of $\{1,\ldots, r\}$. Computing intersections of principal subfields can now be done by simply joining the corresponding partitions of $\{1,\ldots, r\}$, which in practice can be done very quickly and hence, corresponds to a very small percentage of the total CPU time. 

In the remaining of this section we give a description of the principal subfields of $K(t)/K(f(t))$ and in the next section we show how one can compute the partitions associated to every principal subfield of $K(t)/K(f(t))$.

\subsection{Principal Subfields of $K(t)/K(f)$}
In this section we describe the principal subfields of the field extension $K(t)/K(f)$. We follow \cite{vanHoeij2013}, making the necessary changes to our specific case. 

\begin{remark}\label{rem:sepp}
If char($K)=0$, then $\Phi_f$ is separable. If char($K)=p>0$ and $\Phi_f$ is not separable, then $f=\tilde{f}\circ t^{p^s}$, for some $s\geq 1$ and $\tilde{f}\in K(t)$ with $\Phi_{\tilde{f}}$ separable. For this reason, we assume that $\Phi_f$ is separable.
\end{remark}

\begin{definition}\label{def:Li}
Let $F_1,\ldots, F_r$ be the monic irreducible factors of $\Phi_f$ over $K(t)$. For $j=1,\ldots,r$, define the set \begin{equation}\label{eq:Li}
L_j:=\left\{ g(t) \in K(t) \; : \; F_j \mid \Phi_g
 \right\}.
\end{equation}
\end{definition}
If we assume that $F_1=x-t$, then $L_1=K(t)$. Furthermore

\begin{theorem}
Let $F_1,\ldots, F_r$ be the irreducible factors of $\Phi_f$ over $K(t)$. Then $L_1,\ldots, L_r$ are subfields of $K(t)/K(f)$.  
\end{theorem}
\begin{proof}
We show that $L_j$ is closed under multiplication and taking inverse. The remaining properties can be shown in the same fashion.
Let $g(t)=g_n(t)/g_d(t)$ and $h(t)=h_n(t)/h_d(t)$ be elements of $L_j$. By definition, \begin{equation}\label{eq:div1}
F_j\mid \Phi_g \text{ and } F_j\mid \Phi_h.
\end{equation}
Now $g(t)h(t)\in L_j$ if and only if, $F_j \mid \Phi_{gh}$.
By a simple manipulation, one can show that
\begin{equation}
\Phi_{gh}=g_n(x)\Phi_h+h(t)h_d(x)\Phi_g.
\end{equation}
Therefore, by Equation (\ref{eq:div1}), it follows that $F_j\mid \Phi_{gh}$ and hence, $g(t)h(t)\in L_j$. To show that the inverse of $g(t)$ is in $L_j$, notice that \begin{equation}\label{eq:inv}
F_j\mid \Phi_g \text{ if and only if } F_j \mid \Phi_{1/g},
\end{equation} since $\Phi_g = -g(t)\Phi_{1/g}$ in $K(t)[x]$. Therefore, $1/g(t)\in L_j$.
\end{proof}

Finally, we show that the subfields $L_1,\ldots, L_r$ are the principal subfields of $K(t)/K(f)$.

\begin{theorem}
The subfields $L_1,\ldots, L_r$ of $K(t)/K(f(t))$, where $L_j$ is defined as in $(\ref{eq:Li})$, for $j=1,\ldots, r$, are the principal subfields of the extension $K(t)/K(f(t)).$
\end{theorem}
\begin{proof}

Given a subfield $L$ of $K(t)/K(f(t))$, by  L\"uroth's Theorem, there exists a rational function $h(t)\in K(t)$ such that $L=K(h(t))$ and therefore, $f=g\circ h$, for some $g\in K(t)$. By Theorem \ref{theo:main} it follows that $\Phi_h\mid \Phi_f$. Therefore, there exists a set $I_L\subseteq \{1,\ldots, r\}$ such that $\Phi_h = \prod_{i\in I_L} F_i$. We shall prove that \begin{equation}\label{eq:inter}
L = \{ g(t)\in K(t) \; : \; \Phi_h \mid \Phi_g \} = \bigcap_{i\in I_L} L_i.
\end{equation}

Let $g(t)\in K(t)$. Then $g(t)\in L=K(h)$ if and only if $g(t)=\tilde{g}\circ h (t)$, for some $\tilde{g}(t)\in K(t)$, if and only if  $\Phi_h \mid \Phi_g$, by Theorem \ref{theo:main}.
For the second equality, suppose that $g(t)\in \cap_{i\in I_L} L_i$. Then $F_i\mid \Phi_g$, for every $i\in I_L$. Since we are assuming $\Phi_f$ to be separable (see Remark \ref{rem:sepp}), it follows that $\Phi_h = \prod_{i\in I_L} F_i \mid \Phi_g$. Conversely, if $\Phi_h\mid \Phi_g$, then $F_i\mid \Phi_g$, for every $i\in I_L$, that is, $g(t)\in L_i$, for every $i\in I_L$ and hence, $g(t)\in \cap_{i\in I_L} L_i$.
\end{proof}

\section{Partition of Principal Subfields}\label{sec:4}

Let $K(t)/K(f)$ be a separable field extension of finite degree $n$ and let $\Phi_f(x)$ be the minimal polynomial of $t$ over $K(f)$. Let $F_1,\ldots, F_r$ be the irreducible factors of $\Phi_f$ over $K(t)$ and let $L_1,\ldots, L_r$ be the corresponding principal subfields of $K(t)/K(f)$. 

\begin{definition}
A \textit{partition} of $S=\{1,\ldots, r\}$ is a set $\{P^{(1)},\ldots, P^{(s)} \}$ such that $P^{(i)}\subseteq S$, $P^{(i)}\cap P^{(j)} = \emptyset$, for every $i\neq j$ and $\cup P^{(i)}=S$. 
\end{definition}

\begin{definition}
Let $P$ and $Q$ be partitions of $\{1,\ldots, r\}$. We say that $P$ \textit{refines} $Q$ if every part of $P$ is contained in some part of $Q$.
\end{definition}

Recall that $F_1=x-t$. We number the parts of a partition $P=\{ P^{(1)},\ldots, P^{(s)}\}$ in such a way that $1\in P^{(1)}$. Let $P$ be a partition of $\{1,\ldots, r\}$. We say that $P$ is the \textit{finest partition} satisfying some property $X$ if $P$ satisfies $X$ and if $Q$ also satisfies $X$ then $P$ refines $Q$. Moreover, the \textit{join} of two partition $P$ and $Q$ is denoted by $P \vee Q$ and is the finest partition that is refined by both $P$ and $Q$.

 \begin{definition}
Let $F_1,\ldots, F_r$ be the irreducible factors of $\Phi_f$ over $K(t)$. Given a partition $P=\{P^{(1)},\ldots, P^{(s)}\}$ of $\{1,\ldots, r\}$, define the polynomials (so called $P$-\textit{products}) \[g_i:=\prod_{j\in P^{(i)}} F_j\in K(t)[x], \; i=1,\ldots, s.\] 
\end{definition} 

\begin{theorem}[\cite{Szutkoski2016}, Section 2]\label{theo:main2}
Let $f\in K(t)$ and let $F_1,\ldots, F_r$ be the irreducible factors of $\Phi_f$ over $K(t)$. Given a subfield $L$ of $K(t)/K(f)$, there exists a unique partition $P_L=\{P^{(1)},\ldots, P^{(s)}\}$ of $\{1,\ldots, r\}$, called \textit{the partition of} $L$, such that $s$ is maximal with the property that the $P_L$-products are polynomials in $L[x]$. Furthermore, $P_{L \cap L'}=P_L \vee P_{L'}$, that is, the partition of $L\cap L'$ is the join of the partitions $P_L$ and $P_{L'}$ of $L$ and $L'$, respectively.
\end{theorem}

Since $F_1,\ldots, F_r$ are the irreducible factors of $\Phi_f$ over $K(t)$, $P_L$ represents the factorization of $\Phi_f$ over $L$.
 Algorithms for computing the join of two partitions can be found in \cite{Part, Szutkoski2016} (see also \cite{Freese2008}).

Since $1\in P_L^{(1)}$, the first $P_L$-product is the minimal polynomial of $t$ over $L$. 
 As in \cite{Szutkoski2016}, we give two algorithms for computing the partition of the principal subfield $L_i$: one deterministic and one probabilistic, with better performance.

\subsection{A Deterministic Algorithm}

In this section we present a deterministic algorithm that computes, by solving a linear system, the partitions $P_1,\ldots, P_r$ of the principal subfields $L_1,\ldots, L_r$.
We recall (see \cite{Szutkoski2016}, Section 3) that to find the partition of $L_i$ it is enough to find a basis of the vectors $(e_1,\ldots, e_r)\in \{0,1\}^r$ such that $\prod_{j=1}^r F_j^{e_j} \in L_i[x].$ 

\begin{theorem}[\cite{Szutkoski2016}, Lemmas 31 and 32]\label{theo:Fi}
Let $c_1,\ldots, c_{2n}\in K(f)$ be distinct elements and let $h_{j,k}(t):=F'_j(c_k)/F_j(c_k)\in K(t)$. If $(e_1,\ldots, e_r)\in \{0,1\}^r$ is such that $\sum_{j=1}^r e_j h_{j,k}(t)\in L_i$, for $k=1,\ldots, 2n$, then $\prod_{j=1}^r F_j^{e_j}\in L_i[x].$
\end{theorem}

Let us consider $e_1,\ldots, e_r$ as variables. To show that $\sum e_j h_{j,k}(t)\in L_i$ we need an expression of the form $a(t)/b(t)$, where $a,b\in K[t]$.  
Assume $h_{j,k}(t)=n_{j,k}(t)/d_{j,k}(t)$, where $n_{j,k}(t), d_{j,k}(t)\in K[t]$ are coprime. Hence \[ \sum_{j=1}^r e_j \frac{F_j'(c_k)}{F_j(c_k)}=\sum_{j=1}^r e_j h_{j,k}(t)=\sum_{j=1}^r e_j\frac{n_{j,k}(t)}{d_{j,k}(t)}.\]
 
 Furthermore, let $l_k(t)\in K[t]$ be the least common multiple of $d_{1,k}(t),\ldots, d_{r,k}(t)\in K[t]$. Hence \begin{equation}\label{eq:PL}
 \sum_{j=1}^r e_j h_{j,k}(t) = \sum_{j=1}^r e_j\frac{n_{j,k}(t)}{d_{j,k}(t)}=\frac{\sum_{j=1}^r e_j p_{j,k}(t)}{l_k(t)},
 \end{equation}
where $p_{j,k}(t):=l_k(t)\frac{n_{j,k}(t)}{d_{j,k}(t)}\in K[t].$ Hence, $\sum_{j=1}^r e_j h_{j,k}(t) \in L_i$ if, and only if (see Definition \ref{def:Li})
\begin{equation}\label{eq:eval2}
 \left[ \sum_{j=1}^r e_jp_{j,k}(x) - \frac{\sum_{j=1}^r e_j p_{j,k}(t)}{l_k(t)} l_k(x)\right] \bmod F_i =0,
 \end{equation}
where $a \bmod b$ is the remainder of division of $a$ by $b$. By manipulating Equation (\ref{eq:eval2}) we have
\begin{equation}\label{eq:eval3}
 \sum_{j=1}^r e_j \left[ \left(p_{j,k}(x)-h_{j,k}(t)l_k(x) \right) \bmod F_i \right] =0.
 \end{equation}
 
Hence, if $(e_1,\ldots, e_r)\in \{0,1\}^r$ is a solution of (\ref{eq:eval3}), for $k=1,\ldots, 2n$, then Theorem \ref{theo:Fi} tells us that $\prod_{j=1}^r F^{e_j}_j\in L_i[x]$. 

We will now explicitly present the system given by Equation (\ref{eq:eval3}). Let \[q_{j,k}(x):=p_{j,k}(x)-h_{j,k}(t)l_k(x)\in K(t)[x].\] Notice that $\deg_x(q_{j,k})\leq dn$, where $d=\deg_t(c_k)$. Furthermore, let \begin{equation}\label{eq:div}
r_{i,j,k}(x):=q_{j,k}(x) \bmod F_i \in K(t)[x].
\end{equation}
Let $m_j(t)\in K[t]$ be the monic lowest degree polynomial such that $m_j (t) r_{i,j,k}\in K[t][x]$ and let $l\in K[t]$ be the least common multiple of $m_1(t),\ldots, m_r(t)$. Hence
\[ l\sum_{j=1}^r e_j r_{i,j,k}=\sum_{j=1}^r e_j \hat{r}_{i,j,k} \in K[t][x], \]
where $\hat{r}_{i,j,k} = l\cdot r_{i,j,k}\in K[t][x]$. Notice that Equation (\ref{eq:eval3}) holds if and only if, $\sum_{j=1}^r e_j \hat{r}_{i,j,k}=0$. Next, let us write \[ \hat{r}_{i,j,k}= \sum_{d=0}^{d_i-1} \sum_{s=0}^S c_{j}(s,d,k) t^s x^d, \;\;\text{ where } c_{j}(s,d,k)\in K,  \]
where $d_i$ is the degree of $F_i$ and $S\geq 0$ is a bound for the $t$-exponents. Therefore, \[ \sum_{j=1}^r e_j \hat{r}_{i,j,k} = \sum_{d=0}^{d_i-1} \sum_{s=0}^S \left(\sum_{j=1}^r e_j c_{j}(s,d,k)\right) t^s x^d \]
and hence, the system in $e_1,\ldots, e_r$ from Equation (\ref{eq:eval3}) is given by
\begin{equation}\label{eq:system}
 \mathcal{S}_i:=\left\{ \begin{array}{l} 
 \displaystyle \sum_{j=1}^r e_j c_{j}(s,d,k) =0, \begin{array}{ll}  &  d=0,\ldots, d_i-1, \\
  & s=0,\ldots, S, \\
   &  k=1,\ldots, 2n. \end{array}
 
\end{array} 
 \right.
 \end{equation}

\begin{definition}\label{eq:rref}
A basis of solutions $ s_1,\ldots, s_d$ of a linear system with $r$ variables is called a \emph{$\{0,1\}$-echelon basis} if
\begin{enumerate}
\item $s_i=(s_{i,1}, \ldots, s_{i,r})\in \{0,1\}^r$, $1\leq i \leq d$, and
\item For each $j=1,\ldots, r$, there is a unique $i$, $1\leq i \leq d$ such that $s_{i,j}=1.$
\end{enumerate}
\end{definition}

For instance, $S=\{ (1,1,0,0), (0,0,1,0), (0,0,0,1) \}$ is a basis of solutions in $\{0,1\}$-echelon form. If a linear system admits a $\{0,1\}$-echelon basis then this basis coincides with the (unique) reduced echelon basis of this system.

\begin{definition}\label{def:pi}
Let $\mathcal{S}$ be a linear system with $\{0,1\}$-echelon basis $\{s_1,\ldots, s_d\}$. The \textit{partition defined by} this basis is the partition $P=\{ P^{(1)},\ldots, P^{(d)} \}$ where $P^{(j)}=\{i \; : \; s_{j,i}=1\}, \text{ for } j=1,\ldots, d .$
\end{definition}

For instance, $P_S=\{\{1,2\}, \{3\}, \{4\}\}$ is the partition defined by $S$ given above. Therefore, by computing the $\{0,1\}$-echelon basis of the system $\mathcal{S}_i$ given in (\ref{eq:system}) (notice that $\mathcal{S}_i$ admits such basis), the partition defined by this basis is the partition of $L_i$. This is summarized in the next algorithm.

\begin{algorithm}[H]\label{alg:partition}
\caption{\texttt{Partition-D} (Deterministic)}
\begin{algorithmic} 
\STATE \textbf{Input:} The irreducible factors $F_1,\ldots, F_r$ of $\Phi_f(x)$ over $K(t)$ and an index $1\leq i\leq r$.
\STATE \textbf{Output:} The partition $P_i$ of $L_{i}$.
\vspace{0.1cm}
\STATE \text{1. Compute the system $\mathcal{S}_i$ as in (\ref{eq:system}).}
\STATE \text{2. Compute the $\{0,1\}$-echelon basis of $\mathcal{S}_i$. }
\STATE \text{3. Let $P_i$ be the partition defined by this basis.}
\STATE \text{4. \textbf{return} $P_i$.}
\end{algorithmic}
\end{algorithm}

However, algorithm \texttt{Partition-D} is not efficient in practice due to the (costly) $2nr$ polynomial divisions in $K(t)[x]$. We shall present a probabilistic version of this algorithm in Section \ref{sec:4.3}, which allows us to compute $P_i$ much faster.

\subsection{Valuation rings of $K(t)/K$}
In this section we briefly recall the definition and some properties of valuation rings of a rational function field. We will use valuation rings to simplify and speed up the computation of the partition $P_i$ of $L_i$. The results presented in this subsection can be found in \cite{AlgFun}.
\begin{definition}
A \textit{valuation ring} of $K(t)/K$ is a ring $\mathcal{O}\subseteq K(t)$ with the following properties:
\begin{enumerate}
\item $K\subsetneq \mathcal{O} \subsetneq K(t)$, and
\item for every $g\in K(t)$ we have $g\in \mathcal{O}$ or $1/g\in \mathcal{O}$.
\end{enumerate}
\end{definition}

Valuation rings are \textit{local rings}, that is, if $\mathcal{O}$ is a valuation ring, then there exists a unique maximal ideal $\mathcal{P}\subseteq \mathcal{O}$.
\begin{lemma}
Let $p\in K[x]$ be an irreducible polynomial. Let \[ \mathcal{O}_{p}:=\left\lbrace \frac{g_n(t)}{g_d(t)}\in K(t)\; : \; p(x)\nmid g_d(x) \right\rbrace\text{ and } \]
\[ \mathcal{P}_{p}:=\left\lbrace \frac{g_n(t)}{g_n(t)}\in K(t)\; :\; p(x)\nmid g_d(x)\text{ and }p(x) \mid g_n(x) \right\rbrace. \] Then $\mathcal{O}_p$ is a valuation ring with maximal ideal $\mathcal{P}_p$.
\end{lemma}

Furthermore, every valuation ring $\mathcal{O}$ of $K(t)/K$ is of the form $\mathcal{O}_p$, for some irreducible polynomial $p(x)\in K[x]$, or is the place at infinity of $K(t)/K$, that is, $\mathcal{O}=\{\frac{g_n(t)}{g_d(t)}\in K(t)\; : \; \deg(g_n(x))\leq \deg(g_d(x))\}. $

\begin{lemma}\label{lemma:quo}
Let $\mathcal{O}_{p}$ be a valuation ring of $K(t)/K$, where $p\in K[x]$ is an irreducible polynomial, and let $\mathcal{P}_p$ be its maximal ideal. Let $\textbf{F}_p$ be the residue class field $\mathcal{O}_p/\mathcal{P}_p$. Then $\textbf{F}_p \cong K[x]/\left\langle p(x)\right\rangle.$
\end{lemma}

\subsection{A Las Vegas Type Algorithm}
\label{sec:4.3}

In this section we  present a probabilistic version of Algorithm \newline\texttt{Partition-D}. We begin by noticing, as in \cite{Szutkoski2016}, that fewer points are enough to find the partition $P_i$ (usually much less than $2n$). 
Furthermore, the equations of the system $\mathcal{S}_i$ come from the computation of $r_{i,j,k}\in K(t)[x]$ in (\ref{eq:div}), which involves a polynomial division over $K(t)$. Let us define a \textit{good ideal} $\mathcal{P}_p$:

\begin{definition}\label{def:GoodIdeal}
Let $f\in K(t)$ and let $F_1,\ldots, F_r$ be the monic irreducible factors of $\Phi_f$ over $K(t)$. Let $\mathcal{O}_p\subset K(t)$ be a valuation ring with maximal ideal $\mathcal{P}_p$, where $p=p(x)\in K[x]$ is irreducible. Let $\textbf{F}_p$ be its residue field. We say that $\mathcal{P}_p$ is a \textit{good} $K(t)$-\textit{ideal} (with respect to $f$) if
\begin{enumerate} 
\item[1)] $F_i\in \mathcal{O}_{p}[x]$, $i=1,\ldots, r$.
\item[2)] The image of $f$ in $\textbf{F}_p$ is not zero.
\item[3)] The image of $\Phi_f(x)$ in $\textbf{F}_p[x]$ is separable.
\end{enumerate} 
\end{definition}

To avoid the expensive computations of $r_{i,j,k}\in K(t)[x]$, we only compute their image modulo a good $K(t)$-ideal $\mathcal{P}_p$ ( i.e., by mapping $t\rightarrow \alpha$, where $\alpha$ is a root of $p(x)$). These reductions will simplify our computations and we will still be able to construct a system $\tilde{S}_i$ which is likely to give us the partition $P_i$.

\begin{remark}\label{rem:polyp}
Condition $1)$ in Definition \ref{def:GoodIdeal} is equivalent to $p(x) \nmid f_d(x)$ (recall Remark \ref{rem:fact}) and condition $2)$ is equivalent to $p(x)\nmid f_n(x)$. The image of $\Phi_f$ in $\textbf{F}_p[x]$ is separable if $p(t)$ does not divide $R:=resultant(\nabla_f, \nabla_f',x)\in K[t]$. The degree of $R$ is bounded by $(2n-1)n$. Instead of mapping $t \rightarrow \alpha$, we could map $t$ to any element in $\textbf{F}_p=K[x]/\left\langle p(x) \right\rangle$. Hence, if $\text{size}(K)^{d_p} > (2n-1)n$, where $d_p=\deg(p(x))$, then we are guaranteed to find a good evaluation point in $\textbf{F}_p$ which satisfies the conditions in Definition \ref{def:GoodIdeal}. Hence, $d_p\in \mathcal{O}(\log n)$. For best performance, we look for $p(x)$ of smallest degree possible and use the mapping $t\rightarrow \alpha$. Notice that if $\text{char}(K)=0$, we can always choose $p(x)$ linear.
\end{remark}

\subsubsection{Simplified System}

Let $\mathcal{P}_p$ be a good $K(t)$-ideal, where $p=p(x)\in K[x]$ is irreducible. Let $\mathcal{O}_p$ be its valuation ring and $\textbf{F}_p$ be its residue class field. Let $c\in K(f)$ be such that \[h_{j,c}(t):=F'_j(c)/F_j(c)\in \mathcal{O}_{p}\subseteq K(t),\] for $j=1,\ldots, r$, and let $p_{j,c}(t), l_c(t) \in K[t]$ be as in Equation (\ref{eq:PL}). Let $\tilde{F}_i$ be the image of $F_i$ in $\textbf{F}_p[x]$ and let $\tilde{h}_{j,c}$ be the image of $h_{j,c}$ in $\textbf{F}_p$. Let \begin{equation}\label{eq_qtilde}
\tilde{q}_{j,c}(x):=p_{j,c}(x)-\tilde{h}_{j,c} \;l_c(x)\in \textbf{F}_p[x]
\end{equation} and let $\tilde{r}_{i,j,c}:=\tilde{q}_{j,c}(x) \bmod \tilde{F}_i\in \textbf{F}_p[x]$.
Let $d_p$ be the degree of $p(x)\in K[x]$ and let $\alpha$ be one of its roots. By Lemma \ref{lemma:quo} we have $\textbf{F}_p\cong K[\alpha]$ and hence \begin{equation}\label{eq:r_tilde}
 \tilde{r}_{i,j,c} = \sum_{d=0}^{d_i-1} \sum_{s=0}^{d_p-1} C_j(s,d) \alpha^s x^d, \text{ where } C_j(s,d) \in K. 
 \end{equation}
Consider the system $\tilde{\mathcal{S}}_{i,c}$ given by
\begin{equation}\label{eq:redsystem}
 \tilde{\mathcal{S}}_{i,c}:=\left\{ \begin{array}{ll} 
   \displaystyle \sum_{j=1}^r e_j C_{j}(s,d) =0, &  \begin{array}{l} d=0,\ldots, d_i-1, \\
      s=0,\ldots, d_p-1.  \end{array}
\end{array} 
 \right.
 \end{equation}

\noindent where $C_j(s,d)\in K$ is as in Equation \ref{eq:r_tilde}. If $(e_1,\ldots, e_r)\in \{0,1\}^r$ is a solution of $\mathcal{S}_i$, then $(e_1,\ldots, e_r)$ must also satisfy the system $\tilde{\mathcal{S}}_{i,c}$. The converse, however, need not be true. A basis of solutions of $\tilde{\mathcal{S}}_{i,c}$ is not necessarily a basis of solutions of $\mathcal{S}_i$. In fact, a basis of solutions of $\tilde{\mathcal{S}}_{i,c}$ might not even be a $\{0,1\}$-echelon basis. If this happens we need to consider more equations by taking $c' \in K(f)$ such that $h_{j,c'}(t)\in \mathcal{O}_p$, for $j=1,\ldots, r$, and solving $\tilde{S}_i:=\tilde{\mathcal{S}}_{i,c} \cup \tilde{\mathcal{S}}_{i,c'}$, and so on. In subsection \ref{sec:4.3.2} we give a halting condition that tells us when to stop adding more equations to the system $\tilde{S}_i$.

\begin{remark}\label{rem:1}
Advantages of considering $\tilde{\mathcal{S}}_{i}$ over $\mathcal{S}_i$:
\begin{enumerate}
\item Smaller number of polynomial divisions to define $\tilde{\mathcal{S}}_i$.
\item The polynomial divisions are over $K[x]/\left\langle p(x) \right\rangle$, where $p(x)\in K[x]$ is the polynomial defining the ideal $\mathcal{P}$.
\item Smaller system: $\tilde{\mathcal{S}}_i$ has at most $d d_i d_p$ equations, where $d$ is the number of $c$'s used to construct $\tilde{S}_i$, while $\mathcal{S}_i$ has at most $2n d_i S$ equations in $r\leq n$ variables.
\end{enumerate}
\end{remark}

Although in practice we need very few elements $c\in K(f)$ to find $P_i$ (see Table \ref{table}), we were not able to show that $2n$ elements are sufficient to compute $P_i$.

\subsubsection{Halting Condition}\label{sec:4.3.2}

Let $\tilde{S}_i=\cup \tilde{S}_{i,c}$ be a system constructed from several $c\in K(f)$, where $\tilde{S}_{i,c}$ is as in (\ref{eq:redsystem}). 
We will give a halting condition that tells us when to stop adding more equations.
If $\tilde{S}_i$ does not have a $\{0,1\}$-echelon basis then we clearly need more equations. Now let us suppose that $\tilde{S}_i$ has a $\{0,1\}$-echelon basis. Then the partition $\tilde{P}_i$ corresponding to this basis (see Definition \ref{def:pi}) might still be a proper refinement of $P_i$ (the correct partition). To show that $\tilde{P}_i = P_i$ it suffices to show that the $\tilde{P}_i$-products are polynomials in $L_i[x]$.  To do so, we use the following lemma.

  \begin{lemma}[\cite{Szutkoski2016}, Lemma 37]\label{lemma:37}
Let $K$ be a field and $F\in K[x]$ monic and separable. Let $\mathcal{O}\subseteq K$ be a ring such that $ F = g_1 \cdots g_s = h_1\cdots h_s, $ where $g_j, h_j\in \mathcal{O}[x]$ are monic (not necessarily irreducible). Let $\mathcal{P}\subseteq \mathcal{O}$ be a maximal ideal such that the image of $F$ over the residue class field is separable. If $g_j \equiv h_j \bmod \mathcal{P}$, $1\leq j \leq s$, then $g_j = h_j$, $1 \leq j \leq s$.
\end{lemma}

In order to apply this lemma, consider the following map \[\begin{array}{rcc} \Psi_i: K(t) & \rightarrow & K(t,x) \\ 
g(t) & \mapsto & \frac{g_n(x)\bmod F_i}{g_d(x)\bmod F_i}. 
\end{array}
 \] 
Hence, $g(t)\in L_i$ if, and only if, $\Psi_i(g)=g$ (see Definition \ref{def:Li}) and therefore, we can rewrite $L_i=\{ g(t)\in K(t)\; : \; \Psi_i(g(t))=g(t) \}.$
 
 \begin{theorem}\label{theo:correctPi}
 Let $P_i$ be the partition of $L_i$ and let $\tilde{P}_i$ be a refinement of $P_i$. Let $\mathcal{P}_p$ be a good $K(t)$-ideal. If $\tilde{g}_1,\ldots, \tilde{g}_s\in K(t)[x]$ are the $\tilde{P}_i$-products and if $ \Psi_i(\tilde{g}_j)\equiv \tilde{g}_j \bmod \mathcal{P}_p,\; j=1,\ldots, s,  $
 where $\Psi_i$ acts on $\tilde{g}_j$ coefficient-wise, then $\tilde{P}_i = P_i$.
 \end{theorem}
 \begin{proof}
 
Since $\tilde{P}_i$ is a refinement of $P_i$, it suffices to show that the $\tilde{P}_i$-products $\tilde{g}_1\ldots, \tilde{g}_s$ are  polynomials in $L_i[x]$. That is, we have to show that $\Psi_i(\tilde{g}_j)=\tilde{g}_j$, for $j=1,\ldots, s$. Since
\[ \tilde{g}_1\cdots \tilde{g}_s = \Phi_f(x)=\Psi_i(\Phi_f(x))=\Psi_i(\tilde{g}_1)\cdots\Psi_i(\tilde{g}_s) \]
and  $\Psi_i(\tilde{g}_j)=\tilde{g}_j\bmod \mathcal{P}_p$, for $1\leq j \leq s$, then Lemma \ref{lemma:37} implies that $\Psi_i(\tilde{g}_j)=\tilde{g}_j$. Thus $\tilde{g}_j\in L_i[x]$, for $j=1, \ldots, s$, and $\tilde{P}_i=P_i$.
 \end{proof}

This gives us a procedure to determine if the solutions of a system give the partition $P_i$ of the principal subfield $L_i$.

\begin{algorithm}[H]\label{alg:check}
\caption{\texttt{Check}}
\begin{algorithmic} 
\STATE \textbf{Input:} \;\;\;A linear system $\mathcal{S}$ in $e_1,\ldots, e_r$ and an index $i$.
\STATE \textbf{Output:} The partition $P_i$ of $L_i$ or \textit{false}.
\vspace{0.1cm}
\STATE \text{1. Compute a basis of solutions of $\mathcal{S}$.}
\STATE \text{2. \textbf{if} this basis is not a $\{0,1\}$-echelon basis \textbf{then} }
\STATE \text{3.\;\;\;\; \textbf{return} false \;\textit{*Need more equations}.}
\STATE \text{4. Let $\tilde{P}_i$ be the partition defined by this basis.}
\STATE \text{5. Let $\tilde{F}_i$ be the image of $F_i$ in $\textbf{F}_p[x]$.}
\STATE \text{6. Let $\tilde{g}_1,\ldots, \tilde{g}_d$ be the $\tilde{P}_i$-products.}
\STATE \text{7. \textbf{for} every coefficient $c=\frac{c_n(t)}{c_d(t)}\in K(t)$ of $\tilde{g}_1, \ldots, \tilde{g}_d$ \textbf{do}}
\STATE \text{8. \;\;\;\; Let $\tilde{c}$ be the image of $c$ in $\textbf{F}_p$.}
\STATE \text{9. \;\;\;\; \textbf{if} $c_n(x) \bmod \tilde{F}_i \neq \tilde{c}\cdot ( c_d(x)\bmod \tilde{F}_i$) \textbf{then}}
\STATE \text{10. \;\;\;\;\;\;\;\; \textbf{return} false\; \textit{*Need more equations}.}
\STATE \text{11. \textbf{return} $\tilde{P}_i$}
\end{algorithmic}
\end{algorithm}

The correctness of the algorithm follows from Theorem \ref{theo:correctPi}. We end this section by computing the complexity of Algorithm  \rm \texttt{Check}.
\begin{theorem}
One call of Algorithm \normalfont \texttt{Check} \it can be performed with $ \mathcal{O}(n_e r^{\omega-1}+M(n^2)+nM(n)M(d_p))$ field operations, where $d_p$ is the degree of the polynomial defining $\mathcal{P}_p$, $n_e$ is the number of equations in $\mathcal{S}$ and $\omega$ is a feasible matrix multiplication exponent.
\end{theorem}
\begin{proof}
A basis of solutions of $\mathcal{S}$ is computed with $\mathcal{O}(n_e r^{\omega-1})$ field operations. If this basis is not a $\{0,1\}$-echelon basis, then the algorithm returns \textit{false}. The computation of the polynomials $\tilde{g}_1,\ldots, \tilde{g}_d$ in Step 6 can be done with $r-d$ bivariate polynomial multiplications. By Remark \ref{rem:fact}, $\sum \deg_t G_i = \sum \deg_x G_i =n$ and hence, we can compute $\tilde{g}_1,\ldots, \tilde{g}_d$ with $\mathcal{O}(M(n^2))$ field operations (recall that $M( \cdot )$ is super-additive). For each coefficient of $\tilde{g}_1,\ldots, \tilde{g}_d$, we have to verify the condition in Step 9, which can be performed with a reduction modulo $\mathcal{P}_p$ (to compute $\tilde{c}$) and two polynomial divisions over $\textbf{F}_p$. Therefore, for each $c$, we can perform Steps 8 and 9 with $\mathcal{O}(M(n)M(d_p))$ field operations. Since $\sum \deg \tilde{g}_i = n$, we have a total cost of $\mathcal{O}(nM(n)M(d_p))$ field operations for Steps 7-10.
\end{proof}

\subsubsection{Algorithm \texttt{Partitions}}

The following is a Las Vegas type algorithm that computes the partitions $P_1,\ldots, P_r$ of $L_1,\ldots, L_r$.

\begin{algorithm}[H]\label{alg:partition2}
\caption{\texttt{Partitions}}
\begin{algorithmic} 
\STATE \textbf{Input:} The irreducible factors $F_1,\ldots, F_r$ of $\Phi_f$ and a good $K(t)$-ideal $\mathcal{P}_p$ (see Definition \ref{def:GoodIdeal})
\STATE \textbf{Output:} The partitions $P_1,\ldots, P_r$ of $L_1,\ldots, L_r$.
\vspace{0.1cm}
\STATE \text{1. Let $\tilde{\mathcal{S}_i}=\{\; \}$, $i=1,\ldots, r$.}

\STATE \text{2. $I:=\{1,\ldots, r\}$.}
\vspace{0.1cm}
\STATE \text{3. \textbf{while} $I \neq \emptyset$ \textbf{do}}
\STATE \text{4. \;\;\;\; Let $c\in K(f)$ such that $h_{j,c}(t)\in \mathcal{O}_p$, $j=1,\ldots, r$.}
\STATE \text{5. \;\;\;\; Compute $\tilde{q}_{j,c}(x)\in \textbf{F}_p[x]$ as in Equation \ref{eq_qtilde}. 
}
\STATE \text{6. \;\;\;\; \textbf{for} $i \in I$ \textbf{do}}
\STATE \text{7. \;\;\;\;\;\;\;\;\;\; Compute the system $\tilde{\mathcal{S}}_{i,c}$ (see Equation (\ref{eq:redsystem})).}
\STATE \text{8. \;\;\;\;\;\;\;\;\;\; Let $\tilde{\mathcal{S}_i}:=\tilde{\mathcal{S}_i} \cup \tilde{\mathcal{S}}_{i,c}$.}
\vspace{0.1cm}
\STATE \text{9. \;\;\;\;\;\;\;\;\;\; \textbf{if} $\texttt{Check}(\tilde{S}_i, i)\neq\textit{false}$ \textbf{then} }
\STATE \text{10. \;\;\;\;\;\;\;\;\;\;\;\;\;\;\;\; $\texttt{Remove}(I, i)$.}
\STATE \text{11. \;\;\;\;\;\;\;\;\;\;\;\;\;\;\;\; Let $P_i$ be the output of $\texttt{Check}(\tilde{S}_i, i)$.}

\STATE \text{12. \textbf{return} $P_1,\ldots, P_r$.}
\end{algorithmic}
\end{algorithm}

\begin{remark}\label{rem:smallfield}
In general, the elements in Step 4 can be taken inside $K$. This will work except, possibly, when $K$ has very few elements, which might not be enough to find $P_i$. If this happens we have two choices: 
\begin{itemize}
\item[1)] Choose $c\in K(f)\backslash K$ or
\item[2)] Extend the base field $K$ and compute/solve the system $\tilde{\mathcal{S}}_i$ over this extension.
\end{itemize}
We choose the latter. Recall that the solutions we are looking for are composed of $0$'s and $1$'s and hence can be computed over any extension of $K$. Furthermore, extending the base field $K$ does not create new solutions since the partitions are determined by the factorization of $\Phi_f(x)$ computed over $K(t)$, where $K$ is the original field.
\end{remark}

In what follows we determine the complexity of computing $P_1,\ldots, P_r$. We assume, based on our experiments (see Table \ref{table}), that the algorithm finishes using $\mathcal{O}(1)$ elements $c\in K$ (or in a finite extension of $K$) to generate a system $\tilde{\mathcal{S}_i}$ whose solution gives $P_i$.

 \begin{theorem}\label{theo: complexity}
Assuming that \normalfont Algorithm \texttt{Partitions} \it finishes using $\mathcal{O}(1)$ elements inside $K$ in Step 4, the partitions $P_1,\ldots, P_r$, corresponding to the principal subfields $L_1,\ldots, L_r$ of the field extension $K(t)/K(f(t))$, can be computed with an expected number of $ \mathcal{O}( r(rM(n)M(d_p)+M(n^2)))$ field operations, where $d_p$ is the degree of the polynomial defining $\mathcal{P}_p$.
 \end{theorem}
 
\begin{proof}

Given $g=\frac{g_n(t)}{g_d(t)}\in \mathcal{O}_p$, we can compute its image in $\textbf{F}_p$ with $\mathcal{O}(M(deg(g))+M(d_p))$ field operations. Hence, we can compute the images of the polynomials $F_1,\ldots, F_r$ in $\textbf{F}_p$ with $\mathcal{O}(n(M(n)+M(d_p)))$ field operations.

Let $c\in K$. We first compute $h_{j,c}:=F'_j(c)/F_j(c)=G'_j(c)/G_j(c)\in \mathcal{O}_p$, $j=1,\ldots, r$ (see Remark \ref{rem:fact}). Evaluating $G_j\in K[x,t]$ at $x=c$ costs $\mathcal{O}(nd^x_j)$, where $d^x_j=\deg_x(G_j)$. If $d^t_j=\deg_t(G_j)$, simplifying the rational function $G_j'(c)/G_j(c)$ to its minimal form costs $\mathcal{O}(M(d^t_j)\log d^t_j)$. Keeping in mind that $\sum d^t_j = \sum d^x_j=n$, one can compute $h_{j,c}, j=1,\ldots, r$, with $ \mathcal{O}(n^2+M(n)\log n) \text{ field operations}. $

Since $c\in K$, $\deg_t(h_{j,c})\leq d_j^t$ and we can compute the image $\tilde{h}_{j,c}$ of $h_{j,c}, j=1,\ldots, r$, in $\textbf{F}_p$ with $\mathcal{O}(M(n)+rM(d_p))$ field operations. Let us write $h_{j,c}=n_{j,c}/d_{j,c}$, where $n_{j,c}, d_{j,c}$ $\in K[t]$ are coprime. We can compute $l_c=\text{lcm}(d_{1,c}, \ldots, d_{r,c})$ with $r$ $\text{lcm}$ computations, with a total cost of $ \mathcal{O}(rM(n)\log n)$ field operations.
Next, we define $\tilde{q}_{j,c}=p_{j,c}(x)-\tilde{h}_{j,c}(t)l_c(x)$, $j=1,\ldots, r$, where $p_{j,c}(x):=l_c(x)\frac{n_{j,c}(x)}{d_{j,c}(x)}\in K[x]$. The cost of this step is negligible.

For each $i=1,\ldots, r$, to compute the partition $P_i$ we have to compute the system $\tilde{S}_{i,c}$, which involves the division of $\tilde{q}_{j,c}$ by $\tilde{F}_i$, for $j=1,\ldots, r$. Since $\deg(\tilde{q}_{j,c}(x))\leq n$, each of these divisions cost $\mathcal{O}(M(n)M(d_p))$ field operations and hence, we can compute the system $\tilde{S}_{i,c}$ with $ \mathcal{O}(rM(n)M(d_p)) \text{ field operations}. $
This system has at most $d_i d_p$ equations and hence, one call of algorithm \texttt{Check} costs $\mathcal{O}(d_i d_p r^{\omega-1}+M(n^2)+M(n)M(d_p))$. The result follows by adding the complexities and simplifying.
\end{proof} 

\begin{remark}
If Algorithm \normalfont \texttt{Partitions} \it needs $s$ elements $c\in K$ to compute all partitions $P_1,\ldots, P_r$, then the total cost is bounded by $s$ times the cost given in Theorem \ref{theo: complexity}.
\end{remark} 

\begin{corollary}\label{cor:complex}
Let $f\in K(t)$ of degree $n$ and let $F_1,\ldots, F_r$ be the irreducible factors of $\Phi_f(x)\in K(t)[x]$. Let $m$ be the number of subfields of $K(t)/K(f(t))$. One can compute, using fast arithmetic, the subfield lattice of $K(t)/K(f(t))$ with $\tilde{\mathcal{O}}(rn^2)$ field operations plus $\tilde{\mathcal{O}}(mr^2)$ CPU operations.
\end{corollary} 
\begin{proof}
Using fast arithmetic, we can compute the partitions of the principal subfields with $\tilde{\mathcal{O}}(rn^2 d_p)$ field operations, by Theorem \ref{theo: complexity}. By Remark \ref{rem:polyp}, $d_p\in \mathcal{O}(\log n)$. The complete subfield lattice can be computed with $\tilde{\mathcal{O}}(mr^2)$ CPU operations (see \cite{Szutkoski2016}).
\end{proof}

\section{General algorithm and Timings}\label{sec:5}

In this section we outline an algorithm for computing all complete decompositions of $f$ and give an example. Some timings, comparing our algorithm with \cite{Ayad2008}, are also given. 

\subsection{General Algorithm}

Let $f\in K(t)$ and let $F_1,\ldots, F_r$ be the monic irreducible factors of $\Phi_f$. By Theorem \ref{theo:main}, each complete decomposition corresponds to a maximal chain of subfields of $K(t)/K(f(t))$ and vice-versa. Using the algorithms in Section \ref{sec:4} and fast subfield intersection techniques from \cite{Szutkoski2016}, we can (quickly) compute the subfield lattice of $K(t)/K(f(t))$, where each subfield is represented by a partition.
To actually compute the decompositions of $f$, we need to find a \textit{L\"uroth generator} for each subfield. That is, given a partition $P_L$ of $\{1,\ldots, r\}$ representing a subfield $L$, we want to find a rational function $h\in K(t)$ such that $L=K(h)$. 

\begin{theorem}\label{theo:generator}
Let $f\in K(t)$ and let $F_1,\ldots, F_r$ be the monic irreducible factors of $\Phi_f\in K(t)[x]$. Let $L$ be a subfield of $K(t)/K(f)$ and $P=\{P^{(1)}, \ldots, P^{(s)}\}$ be the partition of $L$. Let $g:=\prod_{i\in P^{(1)}} F_i\in L[x]$. If $c\in K(t)$ is any coefficient of $g$ not in $K$, then $L=K(c)$.
\end{theorem}
\begin{proof}
By Lur\"oth's Theorem, there exists a rational function $h(t)\in K(t)$ such that $L=K(h(t))$. Let $\Phi_h\in L[x]$. We may suppose that $\Phi_h\in L[x]$ is the minimal polynomial of $t$ over $L$. Let $g=\prod_{i\in P^{(1)}} F_i\in L[x]$. Since $1\in P^{(1)}$ (recall that $F_1=x-t$), it follows that $g(t)=0$ and hence, $\Phi_h \mid g$. However, $\Phi_h$ and $g$ are monic irreducible polynomials (over $L$) and hence, $g=\Phi_h$. Therefore, $g = h_n(x)-h(t)h_d(x)$. Let $c_i$ be the coefficient of $x^i$ of $g$, then \[c_i=h_{ni}-h(t) h_{d,i}=(-h_{d,i}t+h_{n,i}) \circ h(t), \] where $h_{n,i}$ and $h_{d,i}$ are the coefficients of $x^i$ in $h_n(x)$ and $h_d(x)$, respectively. If $h_{d,i}\neq 0$, then $-h_{d,i}t+h_{n,i}$ is a unit and hence, $L=K(h(t))=K(c_i).$
\end{proof}

Finally, given $f,h\in K(t)$, we want to find $g\in K(t)$ such that $f=g\circ h$. It is known that $g$ is unique (see \cite{ALONSO1995}) and several methods exist for finding $g$. The most straightforward method is to solve a linear system in the coefficients of $g$ (see \cite{Dickerson} for details). Another approach can be found in \cite{Giesbrecht1988} and uses $\mathcal{O}(nM(n)\log n)$ field operations.

\begin{remark}\label{rem:poly}
Our algorithm also works when $f\in K[t]$ is a polynomial if we normalize the generator of each subfield. This follows from Corollary 2.3 of \cite{Ayad2008}. 
If $f=g\circ h$ is a minimal decomposition, then $K(h)$ is a principal subfield and its partition is not refined by any other except $P_1$. Thus, given $P_1,\ldots, P_r$, it is very easy to verify which of these partitions represents a minimal decomposition. For a principal subfield, a L\"uroth generator can be obtained as a byproduct of Algorithm \normalfont \texttt{Check}. \it Hence, given $P_1,\ldots, P_r$, to compute all minimal decompositions of $f$ we only need to compute at most $r-1$ left components. When $\text{char}(K)> 0$, the factorization of $f(x)-f(t)$ can be computed with $\tilde{\mathcal{O}}(n^{\omega+1})$ field operations, where $2< \omega \leq 3$ is a matrix multiplication exponent (see \cite{Bostan2004} and \cite{Lecerf2007}). An algorithm in \cite{Blankertz2014} also computes all minimal decompositions, and take $\tilde{\mathcal{O}}(n^6)$ field operations (for finite fields). For more details, see \cite[Theorem 3.23]{Blankertz2011}.

\end{remark}

\subsection{An Example}

Let $f:=(t^{24}-2t^{12}+1)/(t^{16}+2t^{12}+t^8)$ and consider the extension $\mathbb{Q}(t)/\mathbb{Q}(f)$. The irreducible factors of $\Phi_{f}(x)$ are $F_1=x-t, F_2=x+t, F_3=x+1/t, F_4=x-1/t, F_5=x^2+t^2, F_6=x^2+1/t^2, F_7=x^8+(\alpha/t^4\beta) x^4+1/t^4$ and $F_8=x^8+(\alpha/\beta) x^4+t^4$, where $\alpha=t^8+1$ and $\beta=t^4+1$.

Using Algorithm \texttt{Partitions} we get the following partitions of the principal subfields $L_1,\ldots, L_8$: 
\begin{center}
$\begin{array}{l}
P_1= \{ \{ 1\},\{ 2\},\{ 3\},\{ 4\},\{ 5\},\{ 6\},\{ 7\},\{ 8\}\} \\
P_2= \{ \{ 1,2\},\{ 3,4\},\{ 5\},\{ 6\},\{ 7\},\{ 8\}\}\\
P_3= \{ \{ 1,3\},\{ 2,4\},\{ 5,6\},\{ 7,8\}\}\\
P_4= \{\{ 1,4\},\{ 2,3\},\{ 5,6\},\{ 7, 8\}\}\\
P_5= \{ \{ 1,2,5\},\{ 3,4,6\},\{ 7\},\{ 8\}\}\\
P_6= \{ \{ 1,2,6\},\{ 3,4,5\},\{ 7, 8\}\}\\
P_7= \{ \{ 1,2,5,7\},\{ 3,4,6,8\}\}\\
P_8= \{ \{ 1,2,3,4,5,6,7,8\}\}.\\
\end{array}$
\end{center}

By joining the partitions of all subsets of $\{P_1,\ldots, P_8\}$, we get the following new partitions: 

\begin{center}
$\begin{array}{l}
P_9=P_2 \vee P_4 = \{\{1,2,3,4\}, \{5,6\},\{7,8\}\} \\
P_{10}= P_3 \vee P_6 = \{\{1,2,3,4,5,6\}, \{7,8\}\}.\\
\end{array}$
\end{center}

Hence, $P_1,\ldots, P_{10}$ are the partitions of every subfield of $\mathbb{Q}(t)/\mathbb{Q}(f(t))$. Next we compute all maximal chains of subfields. Recall that the subfield relation translates as refinement of partitions, for instance, $L_5\subseteq L_2$, since $P_2$ refines $P_5$. Therefore, by looking at the partitions $P_1,\ldots, P_{10}$, we see that one maximal chain of subfields is
\[ \mathbb{Q}(f)=L_8\subseteq L_7\subseteq L_5\subseteq L_2 \subseteq L_1=\mathbb{Q}(t).   \]

Now, let us find generators for these fields. As an example, let us find a generator for $L_7$. Following Theorem \ref{theo:generator}, let \[g= \textstyle \prod_{i\in P_7^{(1)}} F_i =F_1F_2F_5F_7 = x^{12}-c x^8-c x^4-1,\] where $c=(t^{12}-1)/(t^8+t^4) $. Since $c\in K(t)\backslash K$, it follows that $L_7=\mathbb{Q}\left( c \right)$. This yields the maximal chain of subfields:
\[ \mathbb{Q}(f)\subseteq \mathbb{Q}\left(c\right) \subseteq \mathbb{Q}(t^4)\subseteq \mathbb{Q}(t^2)\subseteq \mathbb{Q}(t). \]
Finally, we compute the corresponding complete decomposition of $f$ by computing left components. For instance, $\mathbb{Q}(f)\subseteq \textstyle\mathbb{Q}\left( \frac{t^{12}-1}{t^8+t^4}\right)$ implies that there exists $g\in K(t)$ such that $f=g\circ \frac{t^{12}-1}{t^8+t^4}$. In this case we have $g=t^2$ and hence
\[ f=t^2 \circ \frac{t^{12}-1}{t^8+ t^4}. \]
Now $\mathbb{Q}(\frac{t^{12}-1}{t^8+t^4})\subseteq \mathbb{Q}(t^4)$ and we can write $\frac{t^{12}-1}{t^8+t^4}=\frac{t^3-1}{t^2+t}\circ t^4$, and so on. This yields the following complete decomposition: \[ f= t^2 \circ \frac{t^3-1}{t^2+t}\circ t^2 \circ t^2.\]
Doing this for every maximal chain of subfields yields all non-equivalent complete decompositions of $f$.

\subsection{Timings}
Finally, we compare our algorithm \texttt{Decompose}, which returns all non-equivalent complete decompositions of $f$, with the algorithms \texttt{full\_decomp} and \texttt{all\_decomps} from \cite{Ayad2008}, which returns a single complete decomposition and all complete decompositions, respectively. All timings presented below also include the factorization time for $\Phi_f\in K(t)[x]$.

In the table below, $n$ is the degree of $f\in K(t)$ and $r$ is the number of irreducible factors of $\Phi_f$. We also list $d_p$, the degree of the polynomial defining the good $K(t)$-ideal and $\#c$, the number of elements in $K$ (or an extension of $K$, see Remark \ref{rem:smallfield}) used to determine the partitions $P_1,\ldots, P_r$. 

Our algorithm better compares to \texttt{all\_decomps}, since both algorithms return all non-equivalent complete decompositions of $f$. According to our experiments, for \textit{small} values of $r$, the time spent by algorithm \texttt{Decompose} to compute all non-equivalent complete decompositions is similar to the time spent by \texttt{full\_decomp} to compute a single decomposition. However, as $r$ increases, we see a noticeable improvement compared to \texttt{full\_decomp} and more so to \texttt{all\_decomps}. More examples and details about these timings can be found at \url{www.math.fsu.edu/~jszutkos/timings} and the implementation at \url{www.math.fsu.edu/~jszutkos/Decompose}.

\begin{table}[H]
\caption{Timings (in seconds)}
\label{table}
\begin{minipage}{\columnwidth}
\begin{center}
\begin{tabular}{|c|c|c|c|c|c|}
\hline
\multirow{2}{*}{\small $n$} & \multirow{2}{*}{$r$} & \small \multirow{2}{*}{$d_p,\#c$} & \multirow{2}{*}{\small\texttt{Decompose}} & \multicolumn{2}{c|}{\small Ayad \& Fleischmann (2008) \cite{Ayad2008}}        \\ \cline{5-6} 
                     &                                           &                           &                                                                      & \small \texttt{full\_decomp} & \small \texttt{all\_decomps} \\ \hline
12      & 7      &  3,1     & 0.01    & 0.02   & 0.03\\ \hline   
24           & 8      &  1,4     & 0.02    & 0.00   & 0.09 \\ \hline                
144           & 10     &  1,4     & 1.82   & 1.88      & 101.08    \\ \hline
24      & 10     &  3,1     & 0.02    & 0.01   & 0.20 \\ \hline
18         & 12     &  4,1     & 0.05    & 0.06   & 0.81\\ \hline
24      & 14     &  4,1     & 0.07    & 0.51   & 10.57 \\ \hline
60         &  17      &  5,1     &  0.18    & 91.68   & 981.43   \\ \hline
60           & 17      &  1,8     & 0.77    &  485.19  & 4,338.47   \\ \hline
96      & 26     &  2,4  &  0.42    & 211.30 & $>12h$   \\ \hline
60      &  60     &  3,5     & 1.91    &  $> 12h$  & n.a.   \\ \hline
120         & 61       &  3,5      & 2.36 & n.a.     & n.a.  \\ \hline
169         &  91    &  3,7     & 3.41    & n.a.   & n.a.   \\ \hline
120         & 120       &  5,4      & 18.59 & n.a.      &  n.a.   \\ \hline
168        &  168      &  4,9  &  50.53    & n.a.       & n.a.   \\ \hline            
\end{tabular}
\end{center}
\medskip
\footnotesize\emph{n.a.:} not attempted.
\end{minipage}
\end{table}

\bibliographystyle{ACM-Reference-Format}
\bibliography{MyBibired.bib} 

\end{document}